\documentclass[preprint,12pt]{elsarticle}

\usepackage{setspace}
\usepackage[colorlinks,linkcolor=black,citecolor=black,urlcolor=black,%
            bookmarks=false]{hyperref}
\usepackage[ansinew]{inputenc}
\usepackage[british]{babel}
\usepackage[reqno]{amsmath}
\usepackage{amsthm,amsfonts,amssymb}
\usepackage[mathscr]{euscript}
\usepackage{graphicx,color}
\usepackage{multicol,enumerate,fancyhdr,emptypage,xspace,subfig}
\usepackage{makeidx}

\usepackage[inner=30mm,outer=25mm,top=30mm,bottom=30mm]{geometry}

\newcommand{\Z}{{\mathbb Z}}

\newcommand{\im}{\mbox{\rm{Im}}\,}

\newcommand{\F}{{\mathbb F}}

\newcommand{\C}{{\mathcal C}}

\newcommand{\wt}{\mathrm{wt}}

\newcommand{\ie}{\emph{i.e.}}

\newcommand{\rank}{\mbox{rank}\,}

\newcommand{\vect}[1]{\mathbf{#1}}
\newcommand{\vv}{\mathbf{v}}
\newcommand{\uu}{\mathbf{u}}
\newcommand{\ee}{\textbf{e}}

\newcommand{\yy}{\textbf{y}}

\newenvironment{sbmatrix}[1]{\left[\begin{array}{#1}}{\end{array}\right]}

\newcommand{\eg}{\emph{e.g.}}
\newcommand{\FP}{{\mathbb{F}[D]}}
\newcommand{\DD}{(D,D^{-1})} 

\newtheorem{theorem}{Theorem}[section]

\newtheorem{Lemma}[theorem]{Lemma}
\newtheorem{Cor}[theorem]{Corollary}

\usepackage{amssymb}

\makeindex

\begin{document}

\begin{frontmatter}

\title{A new class of convolutional codes and its use in the McEliece Cryptosystem}
\author{P. Almeida and D. Napp }

\address{CIDMA - Center for Research and Development in Mathematics and Applications, Department of Mathematics, University of Aveiro, Portugal palmeida@ua.pt and diego@ua.pt  \\
}

\begin{abstract}
In this paper we present a new class of convolutional codes that admits an efficient algebraic decoding algorithm. We study some of its properties and show that it can decode interesting sequences of errors patterns. The second part of the paper is devoted to investigate its use in a variant of the McEliece cryptosystem. In contrast to the classical McEliece cryptosystems, where block codes are used, we propose the use of a convolutional encoder to be part of the public key. In this setting the message is a sequence of messages instead of a single block message and the errors are added randomly throughout the sequence. We conclude the paper providing some comments on the security. Although there is no obvious security threats to this new scheme, we point out several possible adaptations of existing attacks and discuss the difficulties of such attacks to succeed in breaking this cryptosystem. 
\end{abstract}


\end{frontmatter}

\section{Introduction}

There are two classes of error detecting and correcting codes: block codes and convolutional codes. Block codes divide the information in blocks of a given length and encode them following always the same procedure whereas convolutional codes consider the information as a whole sequence. Despite the fact that convolutional codes are as important for applications as block codes, their mathematical description is much less developed and there exists only a few algebraic constructions of convolutional codes \cite{AlmeidaNappPinto2013,al16,LAGUARDIA2014,Munoz2006,ro99a1}. The achievement of an efficient decoding of convolutional codes under the appearance of errors remains an open fundamental challenge. Maybe the most prominent decoding algorithm is the one given by Viterbi, see also \cite{Lobillo2017}, which applies the principle of dynamic programming to compute the transmitted sequence in a maximum likelihood fashion \cite{Forney73}. The main drawback of this algorithm is that
its complexity increases exponentially with the memory of the code and therefore it becomes computationally infeasible if the degree of the encoder is large. Others decoding procedures like the Massey's threshold decoding algorithm \cite{massey63}, list decoding \cite{jo99,Zigangirov93} or iterative methods \cite{GlHeIg2010,Iglesias2008} have been implemented and even some algebraic decoders have been proposed \cite{ro99}, but the problem of finding more efficient decoding algorithms is very hard in full generality and contains the problem of decoding linear block codes as a special case (known to be an NP-complete problem). \\

In this paper we present a class of convolutional codes that admit a simple iterative algebraic decoding algorithm. The idea of the algorithm seems to be different from the above ideas and is based on building a polynomial encoder as the product of a block code encoder and two invertible Laurent polynomial matrices. Hence, the decoding is performed by inverting the polynomial matrices and decode the corrupted data using the block code. We study this class of convolutional codes and show that the decoding procedure is efficient when the errors are spread out along the received sequence. These features make them good candidates to be applied to the McEliece Public Key Cryptosystem (PKC). In the second part of the paper we study this possibility. \\

The McEliece PKC has received a lot of attention in recent years due to the fact that is one of the most promising post quantum PKC, i.e., able to resist attacks based on quantum computers. Lately, there has been a great interest in post quantum cryptography due to the possibility of the appearance of quantum computers as they would break most of the PKC used in practice, more concretely, all cryptosystems based on factorization and discrete logarithm problems, such as RSA, ElGamal scheme, DSA or ECDSA.
%
Another important advantage of the McEliece cryptosystem is its fast encryption and decryption procedures which require a significantly lower number of operations with respect to alternative solutions (like RSA). However, the original McEliece cryptosystem has two main disadvantages: low encryption rate and large key size, both due to the Goppa block codes it is based on.\\

Motivated by this, there have been several attempts to substitute the underlying Goppa codes by other classes of block codes, \eg , Generalized Reed-Solomon (GRS), Low-Density Parity-Check(LDPC), Quasi-Cyclic, among others. Unfortunately, these alternatives have exposed the system to security threats. A new idea was recently presented in \cite{BBCRS} where Baldi et al. proposed to replace the permutation matrix used in the original McEliece scheme by a more general transformation. This new variant aimed at opening the possibility of trying to use again different classes of codes (\eg  GRS) that were unsuccessfully proposed earlier. Although modifications of this idea and different parameters are currently under investigation, the proposed system in \cite{BBCRS} was broken in \cite{Tillich2015}.\\

Within this thread of research, another interesting variant was recently proposed in \cite{LoJo12} (see also \cite{mogu18}) where the secret code is a convolutional code. 
One of the most appealing features of this system is the fact that its secret generator matrix contains large parts which are generated completely at random and has no algebraic structure as in other schemes such as GRS, algebraic geometry codes, Goppa codes or Reed-Muller codes. Another important property is that convolutional codes allow to deal with long sequence of data in a sequential fashion which makes the computations of large data very efficient. However, the authors already pointed out that this approach suffered from two main problems. The first one being the above mentioned decoding problem, namely, that if one wants a maximum likelihood decoding, such as the Viterbi decoding algorithm, the memory used must be very limited. The second problem stems from the fact that convolutional codes usually start from the all zero state, and therefore the first code symbols that are generated will have low weight parity-checks, which would imply security threats. The scheme considered in \cite{LoJo12} processed just one fixed length block of data and therefore had many similarities with block codes. This allowed the adaptation of existing attacks for block codes and the scheme was finally broken by G. Landais and J-P. Tillich in \cite{Tillich2013}.\\

In this paper we continue this line of work and further investigate the possibility of using convolutional codes instead of block codes. In the proposed scheme the message is not a block vector but a stream of vectors sent in a sequential fashion. Again the security relies in the difficulty of decoding a general convolutional code (specially hard when the degree of the code is large).  Our construction uses large parts of randomly generated matrices in order to mask the secret key and the truncated sliding matrices of the encoders are not full row rank (i.e., they are not generator matrices of a block code). Therefore, the existing attacks to this type of PKC do not seem to be effective. We nevertheless discuss the difficulties of possible adaptations of known attacks to this new context.
It is important to mention that in this work we aim at providing the basic ideas and properties that a convolutional code must have in order to be used in the McEliece cryptosystem. Many important aspects such as setting up the proper parameters to achieve a certain designed work factor or key size requires further investigation and is beyond the scope of the present work.

\section{A new class of Convolutional Codes}

Let $\mathbb F =\mathbb F_q $ be a finite field of size $q$, $\mathbb F((D))$ be the field of formal Laurent series, $\mathbb F\DD$ be the ring of Laurent polynomials, $\F(D)$ is the field of rational polynomials and $\mathbb F[D]$ be the ring of polynomials, all with coefficients in $\mathbb F$. Notice that
\[\mathbb F[D]\subseteq \mathbb F\DD\subseteq \F(D)\subseteq\mathbb F((D)).\]

Unlike linear block codes, there exist several approaches to define convolutional codes. Here we introduce convolutional codes using the  generator matrix approach.  As opposed to block codes, convolutional codes process a continuous sequence of data instead of blocks of fixed vectors. If we introduce a variable $D$, usually called the \emph{delay operator}, to indicate the time instant in which each information arrived or each codeword was transmitted, then we can represent the sequence message $(\uu_{\lambda},\uu_{\lambda + 1} , \cdots  ), \uu_i \in \F^k$ as a Laurent series $\uu\DD= \uu_{\lambda}D^\lambda+\uu_{\lambda + 1} D^{\lambda+1}+ \cdots  \in \F^k((D))$, for $\lambda \in \Z$. In this representation the encoding process of convolutional codes, and therefore the notions of convolutional code and convolutional encoder, can be presented as follows.  \\

A \textit{convolutional code} ${\mathcal C}$ (see definition \cite[Definition 2.3]{mceliece98}) of rate $k/n$ is an $\mathbb{F}((D))$-subspace of $\mathbb{F}((D))^n$ of dimension $k$ given by a rational \emph{encoder matrix} $G(D) \in \mathbb{F}^{k \times n}(D)$,
$${\mathcal C}  =  \im_{\mathbb F((D))}G(D) = \left\{ \uu(D,D^{-1})G(D):\, \uu(D,D^{-1}) \in \mathbb F^{k}((D))\right\}, $$
where $\uu(D,D^{-1}) = \sum_{i\geq \lambda} \uu_i D^i $ is called the {\emph information vector}. If
\[G(D)= \sum_{i=0}^m G_i D^i \in \mathbb{F}^{k \times n}[D]\]
is polynomial, $m$ is called the \emph{memory} of $G(D)$, since it needs to ``remember" the inputs $\uu_i$ from $m$ units in the past. Note that when $m=0$ the encoder is constant and generates a block code. Hence, the class of  convolutional codes generalizes the class of linear block codes in a natural way.



A dual description of a convolutional code ${\mathcal C}$ can be given through one of its  \textit{parity-check} matrices which are  $(n-k)\times n$ full rank rational matrices $H(D)$
such that
\[
\mathcal{C}=\ker_{\mathbb F((D))} H(D)= \left\{ \vect{v}(D) \in \mathbb{F}^{n}((D)) \ \ | \ \
  H(D)\vect{v}(D)=\vect{0} \in \mathbb{F}^{n-k}((D)) \right\}.
\]


%

The ingredients for building up our class of convolutional codes are the following. Let $G \in \F^{k \times n}$ be an encoder of an $(n,k)$-block code that admits an efficient decoding algorithm,

\begin{equation}\label{eq:000}
  T(D,D^{-1})= \displaystyle\sum_{i=-\mu}^{\mu}T_i D^i \in \F^{n \times n}\DD
\end{equation}
an invertible (in $\F\DD$) Laurent polynomial matrix such that
\begin{description}
	\item[a)] its determinant is in $\F$;
	\item[b)] the positions of the nonzero columns of $T_i$ form a partition of $n$;
	\item[c)] each row of $T_i$ has at most one nonzero element, for $i=-\mu, \dots, 0, \dots, \mu$;
\end{description}
and
  \begin{equation}\label{eq:001}
  S(D) = \displaystyle \sum_{i= \mu}^{\nu} S_i D^i \in \F^{k \times k}[D]
 \end{equation}
with $S_\mu\in \F^{k \times k}$ an invertible constant matrix. \\

We propose to construct a convolutional encoder $G'(D)$ as:
\begin{equation}\label{eq:encoder}
  G' (D) = S(D) \ G \ T^{-1}\DD .
\end{equation}

We shall prove that the encoder $G' (D)$ is polynomial and therefore can be used as a blueprint for building a physical encoder for the code, \ie, a device which can be used
to transform $k$ parallel streams of information symbols into $n$ parallel streams of coded symbols. With the goal of developing an algebraic decoding algorithm for $G' (D)$ we start by analyzing the structure of the matrix $T\DD$.

\begin{Lemma}\label{lem:00}
The matrix $T\DD$ is as described above if and only if it can be written as
$$
T\DD=\Pi\,\Delta\DD\,\Gamma,
$$
where
\begin{enumerate}
	\item  $\Gamma\in \F^{n \times n}$ is a permutation matrix;
	\item  $\Delta\DD \in\F^{n \times n}\DD$ is a diagonal Laurent polynomial matrix, whose entries are powers of $D$ ranging from $D^{-\mu}$ to $D^\mu$ with increasing exponents along the diagonal, with $d_i$ entries equal to $D^i$, for $-\mu\leq i\leq\mu$ such that
\begin{equation*}
\mu d_{-\mu}+\cdots+2d_{-2}+d_{-1}=d_1+2d_2\cdots+\mu d_\mu;
\end{equation*}	

\item $\Pi\in \F^{n \times n}$ is an invertible matrix with the following characteristics
\[\Pi=\left [
\begin{array}{c|c|c|c|c}
I_{d_{-\mu}} & U_{1\;2} & \dots &  U_{1\,2\mu} &  U_{1\,2\mu+1} \\ \hline
U_{2\;1} & I_{d_{-\mu+1}} & \dots &  U_{2\,2\mu} & U_{2\,2\mu+1} \\ \hline
\vdots & \vdots & \ddots & \vdots & \vdots \\ \hline
U_{\mu+1\;1} & \dots &  I_{d_0} & \dots & U_{\mu+1\,2\mu+1} \\ \hline
\vdots & \vdots & \ddots & \vdots & \vdots\\ \hline
U_{2\mu+1\;1} & U_{2\mu+1\;2} & \dots & U_{2\mu+1\;2\mu} & I_{d_\mu}
\end{array}\right ]\]
where each $U_{i\;j}$ has at most one nonzero entry in each row.
\end{enumerate}
\end{Lemma}
\begin{proof}
Let $T(D,D^{-1})= \sum_{i=-\mu}^{\mu}T_i D^i \in \F^{n \times n}\DD$ be a Laurent polynomial matrix such that the positions of the nonzero columns of $T_i$ form a partition of $n$ and each row of $T_i$ has at most one nonzero element, for $i=-\mu, \dots, 0, \dots, \mu$. It follows that for each column of $T(D,D^{-1})$, its nonzero entries have equal exponents of $D$. Reorder its columns so that the exponents of $D$ are in increasing order along the columns (this corresponds to multiply $T$ on the right by a permutation matrix $\Gamma'$).

Consider the diagonal Laurent polynomial matrix $\Delta'\DD$ whose $(i,i)$ entry is $D^j$ if the nonzero elements of the column $i$ of $T\DD \,\Gamma'$ are of the form $aD^{-j}$, where $a\in\F$ and $-\mu\leq j\leq \mu$. Then the matrix $\Pi=T\DD \,\Gamma'\,\Delta'\DD$ satisfies the conditions of $3$. Therefore, $T\DD=\Pi\,\Delta'\DD^{-1}\,(\Gamma')^T$, and since the determinant of $T\DD$ is in $\F$,  $(\Delta'\DD)^{-1}$ satisfies condition $2$. It is also easy to see that any matrix that is the product of matrices satisfying the above conditions is of the form stated in (\ref{eq:000}) and satisfies the conditions a), b) and c).
\end{proof}

\begin{Cor}
The inverse of $T^{-1}\DD$ is as follows:
\begin{equation*}
P\DD = T^{-1}\DD= \displaystyle\sum_{i=-\mu}^{\mu}P_i D^i \in \F^{n \times n}\DD,
\end{equation*}
where the positions of the nonzero rows of $P_i$ form a partition of $n$.  Consequently, $G'(D)$ as in (\ref{eq:encoder}) is polynomial and has memory $\leq \mu+\nu$.
\end{Cor}

\begin{proof}
Suppose that $T\DD=\Pi\,\Delta\DD \,\Gamma$, then $P\DD=\Gamma^T\Delta\DD^{-1}\Pi^{-1}$. Clearly, for each row of $\Delta\DD^{-1}\Pi^{-1}$ all nonzero entries have equal exponents of $D$ and when we multiply by the permutation matrix we only interchange the rows. Therefore, the positions of the nonzero rows of $P_i$ form a partition of $n$. The last statement follows easily from previous considerations.
\end{proof}

Let $\wt(\vv\DD)$ be the Hamming weight  of a polynomial vector $
\vv\DD=\sum\limits_{i\in \Z} \vv_i D^i,
$ defined as $\wt (\vv\DD)=\sum\limits_{i\in \Z} \wt(\vv_i),$ being $\wt(\vv_i)$ the number of the nonzero components of $\vv_i\in \mathbb{F}^n$.

For the sake of simplicity we shall consider information vectors that start at time instant zero and have finite support, \ie, $\uu (D)\in \F^k[D]$ is polynomial. We will also consider the error vectors $\ee (D)\in \F^n[D]$ to be polynomial.

\begin{Lemma}\label{lem:01}
Let $T\DD$ be as described above and $\ee (D)= \displaystyle \sum_{i\geq 0} \ee_i D^i \in \F^n[D]$, a random error vector satisfying
  \begin{equation}\label{eq:-2}
  \wt((\ee_ {i}, \ee_{i+1}, \dots, \ee_{i+2\mu}))\leq t
  \end{equation}
  for all $i\geq 0$. Then, all the coefficients of $\ee(D) T\DD$ have weight less than or equal to $t$.
\end{Lemma}

\begin{proof}
  It readily follows from the observation that if each row of $T_i$ has at most one nonzero element, then $\wt (\ee_j) \geq \wt (\ee_j T_i)$ for all $j \geq 0$ and $-\mu\leq i\leq \mu$. Take $\ee_j=0$ for $j<0$, then the $s$-th coefficient of $\ee(D) T\DD$ is given by $\displaystyle\sum_{i=-\mu}^{\mu}\ee_{s-i} T_i$, for $s\geq -\mu$, and the result follows.
\end{proof}

The condition in (\ref{eq:-2}) describes the maximum number of errors allowed within a time interval and is similar to the sliding window condition introduced in \cite{ba15b} to describe the possible error patterns that can occur in a given channel.

\begin{theorem}\label{th:01}
Let $G' (D)$ be the encoder as described in (\ref{eq:encoder}), $t$  the correcting error capability of $G$, $\uu(D)$ the information sequence and $\ee(D)$ an error vector satisfying (\ref{eq:-2}). Then, the received data $\yy(D)= \uu(D) G' (D) +   \ee(D)\in \F^n[D]$ can be successfully decoded.
\end{theorem}

\begin{proof}
Let $\yy(D)= \uu(D) G' (D) +   \ee(D)\in \F^n[D]$ where $\uu(D)= \uu_{0} +\uu_{ 1} D+ \cdots + \uu_{\ell} D^\ell$ is an information sequence and $\ee(D)=\ee_{0} +\ee_{ 1} D+ \cdots + \ee_{\ell+\mu+\nu} D^{\ell+\mu+\nu}$ is an error vector satisfying (\ref{eq:-2}).

Multiplying $\yy(D)$ by $T\DD$ from the right yields the polynomial equation
$$
\yy(D) T\DD =\uu(D) S(D) G + \ee (D) T\DD.
$$
Hence, each coefficient of this polynomial is of the form $\widehat{\uu}_i G + \widehat{\ee}_i$ where
\[\sum_{i=\nu }^{\nu + \ell  } \widehat{\uu}_i D^i =\uu (D) S(D)\]
and
\[\sum_{i\geq -\mu}^{} \widehat{\ee}_{i}D^i = \ee(D) T\DD.\]
By Lemma \ref{lem:01} it follows that  $wt(\widehat{\ee} )\leq t$ and therefore each $\widehat{\uu}_{i}$ can be recovered. Further, one can retrieve $\uu(D)$ from
$$\uu (D) S(D)= \sum_{i=\mu}^{\ell+\nu } \widehat{\uu}_i D^i $$
as
\begin{eqnarray*}
\begin{sbmatrix}{cccc}
    \widehat{\uu}_{ \mu} &    \widehat{\uu}_{\mu  +1} &    \cdots &    \widehat{\uu}_{\ell +\nu }
  \end{sbmatrix} = & \\
  =
  \begin{sbmatrix}{cccc}
    \vect{u}_{0} &    \vect{u}_{ 1} &    \cdots &    \vect{u}_{\ell}
  \end{sbmatrix} &
  \begin{sbmatrix}{cccccccc}
    S_{\mu}   & S_{\mu +1 } & \cdots & S_{\nu} & &    \\
      & S_{\mu} & S_{\mu +1}  & \cdots & S_{\nu} &   \\
   & & \ddots &   \ddots  &   \ddots & \ddots  \\
   & &        & S_{\mu} & \cdots & & S_{\nu} \\
   & &        &         & \ddots & & \vdots    \\
   & &        &         &      & \ddots & \vdots   &  \\
   & &        &         &      &  & S_{\mu} &    \\
  \end{sbmatrix},
\end{eqnarray*}
\begin{equation*}
\hspace*{4.5cm} \underbrace{\hspace*{6.2cm}}_{=:S_{truc}(\ell)\in \F^{(\ell + \lambda +1) k \times (\ell + \lambda +1) k}}
\end{equation*}
and the matrix $S_{truc}(\ell)$ is invertible as $S_\mu$ is invertible.
\end{proof}


If we denote
\[\sum_{i\geq -\mu} \widehat{\yy}_i D^i = \yy (D) T\DD , \]
then, at each time instant $j$, with $j\geq -\mu$, we compute each
\[\widehat{\yy}_j  = \sum_{i=-\mu}^{\mu}\yy_{j-i}  T_i\]
by performing at most $2 \mu +1$ multiplications of vectors of size $n$ by a matrix of order $n$, then decode $\widehat{\yy}_i$ using the decoding algorithm corresponding to $G$ to obtain $\widehat{\uu}_j$. Finally we retrieve $\uu_j$ by performing at most $2 \mu +1$ multiplications of vectors of size $n$ by a matrix of order $n$ and $2\mu$ sums of vectors of order $n$. Recall that the complexity of each vector-matrix product is $O(n^2)$, so the complexity of decoding is $O(n^2\mu)$.

%

\section{A new variant of the McEliece PKC based on convolutional codes}

\subsection{The original McEliece cryptosystem}

Next we briefly recall the original McEliece PKC. Let $G\in \F^{k \times n}$ be an encoder of an $(n,k)$-block code $\C$ capable of correcting $t$ errors, $S\in \F^{k \times k}$ an invertible matrix and $P\in \F^{n \times n}$ a permutation matrix. In the classical McEliece cryptosystem $G$, $S$ and $P$ are kept secret, but $G' = SGP$ and $t$ are public. Bob publishes $G'$ and Alice encrypts the cleartext message $\textbf{u} \in \F^k$ to produce $\textbf{v}=\textbf{u} G'$, chooses a random error $\ee \in \F^n$ with weight $\wt (\ee)\leq t$ and sends the ciphertext
 $$
 \textbf{y} = \textbf{v} + \textbf{e} = \textbf{u}G' + \textbf{e}= \textbf{u}SGP + \textbf{e} .
 $$

The generator matrix $G$ is selected in such a way that allows an easy decoding so that when Bob receives the vector $\textbf{y}$, multiplies from the right by the inverse of $P$ and recovers $\textbf{u}S$ by decoding $(\textbf{u}S) G + \textbf{e} P^{-1}$ as $wt(\textbf{e} P^{-1})\leq t$. Finally, Bob multiplies $\textbf{u}S$ by the matrix $S^{-1}$ to obtain $\textbf{u}$.

It is important to recall that the security of this cryptosystem lies in the difficulty of decoding a random encoder, known to be an NP hard problem. Hence, $G$ needs to admit an easy decoding algorithm but $G'$ has to look as random as possible. 

\subsection{A new variant using convolutional codes}

Here we propose a new scheme of the McEliece PKC where a secret encoder of a block code is masked by polynomial matrices yielding a polynomial encoder of a convolutional code, which constitutes the public key. We shall consider the type of encoders described in Section 2. In this context,  the information vector $\uu(D,D^{-1})$ represents the key to be interchanged and thus will be a polynomial with designed fixed degree.

%

Let $S(D)$, $G$, $P\DD = T^{-1}\DD$ and $t$ be as in Section 2. The proposed scheme works as follows:\\

\textbf{Secret key}: $\{ S(D) , G, P\DD \}$.\\

\textbf{Public key}: $\{ G' (D) = S(D) G P\DD ,\ $t$ \}$.\\

\textbf{Encryption:} Alice selects an error vector $\ee(D)$ satisfying (\ref{eq:-2}) and encrypts the message $\uu(D)= \uu_0 + \uu_1 D + \uu_2 D^2 + \cdots +\uu_\ell D^\ell \in \FP^k$ as
        \begin{equation}\label{eq:00}
        \yy (D)= \uu(D) G' \DD +   \ee (D).
        \end{equation}

\textbf{Decryption:} Bob multiplies (\ref{eq:00}) from the right by the matrix $T\DD$  to obtain
  \begin{equation}\label{eq:01}
    \uu(D) S(D) G + \ee (D) T\DD,
  \end{equation}
decode using $G$ and finally recover the message $\uu (D)$  from $\uu(D) S(D)$.\\

It is important to note that the scheme presented is possible due to Theorem \ref{th:01} and can be implemented \emph{sequentially}, \ie, we can encrypt the data $\uu_i$'s as it enters into the encoder $G'$ and decrypt the sequence of $\yy_i$'s as they arrive, without waiting for the end of the transmission.


\section{Comments on the security}

In this section we outline possible attacks to the proposed cryptosystem and show how the properties analyzed in Section 2 provide key security. As mentioned before a thorough study of the security requires further investigation. The ideas presented here are given in general terms and setting the right parameters to provide a desired work factor, key size or an certain encryption/decryption complexity lies beyond the scope of the work. \\

There are two main classes of attacks to the McEliece cryptosystem. One, called structural attacks, is based on specific techniques aiming at exploiting the particular structure of the code. The other class, called Information Set Decoding (ISD), try to decode the ciphertext in order to obtain the plaintext directly from the public key.

Each one of these classes will be analysed considering two different situations: One when the attacker deals with the whole message and another when  only truncated parts (intervals) of the sequence that forms the message are considered. This is explained next. \\

The ciphertext is generated as $\yy(D)= \uu(D) G' (D) +   \ee(D)$  or equivalently,
\begin{equation}\label{vecty}
\begin{sbmatrix}{cccccc}
    \vect{y}_{0} &    \vect{y}_{1} &    \cdots &    \vect{y}_{\ell} & \cdots & \vect{y}_{\ell + \nu+\mu}
  \end{sbmatrix}
\end{equation}
is equal to the multiplication of
\begin{eqnarray}\label{equ}
  \begin{sbmatrix}{cccc}
    \vect{u}_{0} &    \vect{u}_{1} &    \cdots &    \vect{u}_{\ell}
  \end{sbmatrix}
\end{eqnarray}
with

\begin{eqnarray}\label{eqGlinha}
  \begin{sbmatrix}{ccccccccc}
    G'_{0}   & G'_{1} & \cdots & G'_{\nu+\mu} & &  &  \\
      & G'_{0} & G'_{1}  & \cdots & G'_{\nu+\mu} &  &  \\
   & & \ddots &   \ddots  &    & \ddots  \\
   & &        & G'_{0} & G'_{1} & \cdots & G'_{\nu+\mu} \\
   & &        &         & \ddots & \ddots & & \ddots  \\
   & &        &         &        & G'_{0} & G'_{1} & \cdots & G'_{\nu+\mu}\\
  \end{sbmatrix}
\end{eqnarray}
and adding the error vector
\begin{eqnarray}
    \begin{sbmatrix}{cccc}
    \vect{e}_{0} &    \vect{e}_{1} &    \cdots &    \vect{e}_{\ell+ \nu+\mu}
  \end{sbmatrix}.
\end{eqnarray}

%
%

An attacker could consider the full row rank matrix in (\ref{eqGlinha}) as the generator matrix of a block code. Note, however, that the size of this block code is $k (\ell +1) \times n(\ell +1 +\nu + \mu) $ and therefore it will be very large for large values of $\ell$ and $\nu$, even if $n$, $k$ and $\mu$ are small. Notice that to decode Bob uses just matrices of size $k\times n$.\\

Another possibility is to consider an interval of the ciphertext sequence in (\ref{vecty}) instead of the whole sequence. However, convolutional codes have \emph{memory} and therefore each $\vect{y}_{i}$ depends on previous data. In order to discover $\vect{y}_{i}$ one needs to compute previous $\vect{y}_{j}$, i.e., one must estimate the state of the convolutional code at time instant $i$.

Thus, it seems natural to try to attack the first vectors $[ \vect{y}_{0}     \vect{y}_{1}     \cdots     \vect{y}_{s}]$ for different values of $s$, as the initial state is assumed to be zero. Indeed, this possible weakness when using convolutional codes in the McEliece cryptosystem was already pointed out in \cite{LoJo12}. 

For these reasons we are particularly interested in the security of the following interval of data
\begin{equation}\label{eq:011}
    \begin{sbmatrix}{cccc}
    \vect{u}_{0} &    \vect{u}_{1} &    \cdots &    \vect{u}_{s}
  \end{sbmatrix}
  \begin{sbmatrix}{cccccc}
    G'_{0}   &  G'_{1} & \cdots & G'_{s} \\
      & G'_{0} &    \cdots & G'_{s-1}     \\
   & & \ddots &   \vdots      \\
   & &        & G'_{0}  \\
  \end{sbmatrix}
+  \begin{sbmatrix}{cccc}
    \vect{e}_{0} &    \vect{e}_{1} &    \cdots &    \vect{e}_{s}
  \end{sbmatrix}
  =
  \begin{sbmatrix}{cccc}
    \vect{y}_{0} &    \vect{y}_{1} &    \cdots &    \vect{y}_{s}
  \end{sbmatrix},
\hspace*{3.14cm}\underbrace{\hspace*{4.2cm}}_{=:G'_{truc}(s)}
\end{equation}
$s\leq \ell$ and obviously $G_i=0$ for $i>\nu +\mu$. In practice, block codes typically consider large values of $n$ and $k$ whereas convolutional codes are used with small values of $n$ and $k$ and the memory is tuned for the application at hand. Note that $\ell$ can be taken to be very large without increasing the size of the public key nor the operations performed at each time instant.

\subsection{Structural attacks}

One general observation about the security of the proposed cryptosystem is about the way the public key is generated. As opposed to previous constructions of the variants of the McEliece PKC, $G' (D)=S(D) G P\DD$ is constructed using large parts randomly generated, and this makes, a priori, structure attacks more difficult. Effectively, the coefficients of  $S(D)$ are totally randomly generated except for the $S_\mu$ that is required to be invertible. Moreover, we will show next that the set of admissible $P\DD$ is considerably large, and therefore  $P\DD$ also introduces additional randomness into the system.\\


\begin{Lemma}
The number of possible $S(D)$ is equal to
\[(\nu-\mu)q^{k^2}\prod_{j=0}^{k-1}(q^k-q^j).\]
\end{Lemma}
\begin{proof}
Since $S_\mu$ is invertible there are $\prod_{j=0}^{k-1}(q^k-q^j)$ possibilities for S$_\mu$. There is no restriction for $S_i$ with $\mu<i\leq \nu$, so there are $q^{k^2}$ possibilities for each $S_i$.
\end{proof}

To estimate the cardinality of the set of possible $T\DD$ in our cryptosystem recall that we may take $T\DD=\Pi\,\Delta\DD\,\Gamma$, where $\Gamma$ is any permutation matrix of order $n$,
\[\Delta\DD=\left [
\begin{array}{cccccccccccc}
1/D^\mu & & & & & & & & & & & \\
& \ddots & & & & & & & & & &\\
& & 1/D^\mu  & & & & & & & & & \\
& & & \ddots & & & & & & & &\\
& & & & 1 & & & & & & & \\
& & & & & \ddots & & & & & \\
& & & & & & 1 & & & & \\
& & & & & & & \ddots & & & \\
& & & & & & & & D^\mu & & \\
& & & & & & & & & \ddots & \\
& & & & & & & & & & D^\mu
  \end{array}\right ]\]
\begin{equation*}
\hspace*{2.0cm}\underbrace{\hspace*{2.5cm}}_{d_{-\mu}}\hspace*{1.2cm} \underbrace{\hspace*{1.2cm}}_{d_0}\hspace*{1.5cm} \underbrace{\hspace*{2.5cm}}_{d_\mu}
\end{equation*}
where $\mu d_{-\mu}+\cdots+2d_{-2}+d_{-1}=d_1+2d_2\cdots+\mu d_\mu$, and
\[\Pi=\left [
\begin{array}{c|c|c|c|c}
I_{d_{-\mu}} & U_{1\;2} & \dots &  U_{1\,2\mu} &  U_{1\,2\mu+1} \\ \hline
U_{2\;1} & I_{d_{-\mu+1}} & \dots &  U_{2\,2\mu} & U_{2\,2\mu+1} \\ \hline
\vdots & \vdots & \ddots & \vdots & \vdots \\ \hline
U_{\mu+1\;1} & \dots &  I_{d_0} & \dots & U_{\mu+1\,2\mu+1} \\ \hline
\vdots & \vdots & \ddots & \vdots & \vdots\\ \hline
U_{2\mu+1\;1} & U_{2\mu+1\;2} & \dots & U_{2\mu+1\;2\mu} & I_{d_\mu}
\end{array}\right ]\]
where each $U_{i\;j}$ has at most one nonzero entry in each row. If $i>j$, the matrices $U_{i\,j}$ should be chosen so that $\Pi$ is invertible and each row has at most one nonzero element. The simplest case is when they are all null matrices. In this case, there are
\[(q-1)(d_{j-\mu-1}+1)^{d_{i-\mu-1}}\]
possible $U_{i\;j}$ matrices, when $i<j$. There are also $n!$ possible $\Gamma$ matrices.

\begin{Lemma}
The number of possible $\Delta\DD$ matrices such that $\Delta\DD\neq I_n$, is
\[\sum_{r=1}^{\mu\lfloor n/2\rfloor}\sum_{i=1}^{\min(r,n)}p(r,i,\mu)\left (\sum_{j=1}^{\min(r,n-i)}p(r,j,\mu)\right )\]
where $p(r,i,\mu)$ is the number of partitions of $r$ into exactly $i$ parts and in which the largest part has size at most $\mu$.
\end{Lemma}
\begin{proof}
The number of possible $\Delta\DD$ matrices is equal to the number of $2\mu+1$-tuples $(d_{-\mu}, \dots, d_0, \dots, d_\mu)$ such that $d_{-\mu}+\cdots+ d_0+\cdots+ d_\mu=n$ and
\[\mu d_{-\mu}+\cdots+2d_{-2}+d_{-1}=d_1+2d_2\cdots+\mu d_\mu.\]
First notice that, if $\mu d_{-\mu}+\cdots+2d_{-2}+d_{-1}=0$ then $d_0=n$ and $\Delta\DD=I_n$ and that under those two conditions, we cannot have $\mu d_{-\mu}+\cdots+2d_{-2}+d_{-1}>\mu\lfloor n/2\rfloor$.

Let $0\leq r\leq \mu\lfloor n/2\rfloor$. The number of $\mu$-tuples $(d_{-\mu}, \dots, d_1)$ such that $\mu d_{-\mu}+\cdots+d_{-1}=r$ is equal to
\[\sum_{i=1}^{\min(r,n)}p(r,i,\mu),\]
where $p(r,i,\mu)$ is the number of such $\mu$-tuples with $d_{-\mu}+\cdots+ d_{-1}=i$, for $1\leq i\leq \min(r,n)$. But if $d_{-\mu}+\cdots+ d_{-1}=i$, then $d_1+\cdots, d_\mu\leq n-i$. In order to obtain a $2\mu+1$-tuples $(d_{-\mu}, \dots, d_0, \dots, d_\mu)$ we also need that $d_1+2d_2\cdots+\mu d_\mu=r$, hence the result follows.
\end{proof}

Note that the partitions $p(r,i,\mu)$ can be computed by recurrence \cite{hardy75}. Next we look at possible structural attacks to the truncated sliding generator matrix $G'_{trunc}(s)$. Recall that $\uu (D) S(D)= \sum_{i=\mu }^{\nu + \ell  } \widehat{\uu}_i D^i $. Equation (\ref{eq:011}) can be written as

\begin{equation}\label{eq:trucatedsliding}
    \begin{sbmatrix}{cccc}
        \widehat{\uu}_{\mu} &    \widehat{\uu}_{\mu  +1} &    \cdots &    \widehat{\uu}_{\mu +s }
    \end{sbmatrix}
    \begin{sbmatrix}{cccccc}
    G   &   &  &  \\
      & G &     &      \\
   & & \ddots &   \vdots      \\
   & &        & G  \\
    \end{sbmatrix}
\begin{sbmatrix}{ccccccccc}
    P_{-\mu} & \cdots   & P_{0} & \cdots & P_{\mu} & &  & \\
      & P_{-\mu}  & \cdots & P_{0} & \cdots & P_{\mu} &  &  \\
   & & \ddots & &  \ddots & &\ddots   & \\
   & & & \ddots & & \ddots & &  P_{\mu}      \\
   & & & & \ddots & & \ddots &      \\
   & & & & & \ddots & & P_0\\
   & & & & & & \ddots & \\
   & & & & & & &  P_{-\mu}
  \end{sbmatrix}
\end{equation}
\begin{equation*}
    \hspace*{8.14cm}\underbrace{\hspace*{5.2cm}}_{=:P_{truc}(s)}
\end{equation*}

As all the matrices in the equation (\ref{eq:trucatedsliding}) are constant, then it may seem to describe a similar situation as the original McEliece PKC using block codes. Note, however, that $P_{truc}(s)$ is neither a permutation nor an invertible matrix. The idea of using more general transformations rather than permutations was recently investigated in \cite{BBCRS}. This scheme was broken in \cite{Tillich2015} using a structural attack based on the notion of square code construction. The matrices $P_{truc}(s)$ are not invertible which make the attack in \cite{Tillich2015} inefficient in this case.

\subsection{Plaintext recovery}

These attacks try to decode a random linear code without requiring any knowledge of the secret key. However, trying to decode directly a convolutional code using ML decoding, e.g., the Viterbi decoding algorithm, seems very difficult as the complexity grows exponentially with the memory of the code. 

The plaintext recovery type of attack is typically performed using information set decoding algorithms (ISD). Information set decoding tries to solve the following NP-hard problem: Given a random looking generator matrix $G'\in\F^{k \times n}$ of a linear code $C'$ and a vector $\yy= \textbf{u}G' + \textbf{e}$, recover $\uu$. Roughly speaking, the problem is that of decoding a random linear code. The first step of any ISD algorithm is to find a size-$k$ index set $I\subset \{1,2,\dots,n\}$ such that the sub-matrix of $G'$ with the columns indexed by $I$ forms an invertible matrix of order $k$, or, equivalently, the sub-matrix of the parity check matrix $H'$ (associated with $G'$) with the columns indexed by $I^\star=\{1,2,\dots,n\}\setminus I$ forms an invertible matrix of order $n-k$. The set $I$ is called an {\em Information Set}.

The second step depends of the algorithm we are using, but the basic idea is to guess the $I$-indexed part $e_I$ of the error vector $\ee$ according to a predefined method (that depends on each specific algorithm) and try to obtain the whole $\ee$ from these assumptions. For example:

\begin{itemize}
	\item In Prange algorithm (also called plain ISD algorithm) we guess that $\wt(e_I)=0$.
	\item In Lee-Brickell algorithm we guess that $\wt(e_I)=p$, for a fixed value $p$.
	\item In Stern algorithm we separate $I$ in two sets $I_1$ and $I_2$ with approximately the same size and guess that $\wt(e_{I_1})=p/2$,  $\wt(e_{I_2})=p/2$, and so  $\wt(e_{I})=p$. Stern considers also another restrition in the $I^\star$-indexed part of $\ee$, namely that the first $m$ coordinates have zero weight.
\end{itemize}

Next, we analyse how these ideas could be adapted to our context. Consider the full row rank matrix in (\ref{eqGlinha}) and ISD attacks using, for instance, the Stern algorithm (or some of its more improved versions \cite{Peters2010,Umana09,BeSi2018}). The probability of success is
\[\frac{\left (\begin{array}{c}
k(\ell+1)/2\\
p/2
\end{array}\right )^2\left (\begin{array}{c}
n(\ell+\mu+\nu+1)-k(\ell+1)-m\\
t-p
\end{array}\right )}{\left (\begin{array}{c}
n(\ell+\mu+\nu+1)\\
t
\end{array}\right )}\]
and the search time is
\[\left (\begin{array}{c}
	k(\ell+1)/2\\
	p/2
\end{array}\right )pm[\frac{\left (\begin{array}{c}
	k(\ell+1)/2\\
	p/2
	\end{array}\right )^2}{2^m}(p(n(\ell+\mu+\nu+1)-k(\ell+1)-m)+O(1)).\]
As these values depend on $\ell$ if one chooses a large enough $\ell$ the work factor will achieve a certain designed value. If one considers an interval as in (\ref{eq:011}) the problem is that this is not an standard decoding problem as $G'_{truc}(s)$ is not a full row rank matrix. One can try to adapt the ISD algorithms to this situation. One way to do it is instead of starting with size-$k$ information sets, we look for smaller sets (that will depend on the number of nonzero rows each $P_i$ has).\\

Let denote $k_s= \rank (G'_{truc}(s))$ and select a size-$k_s$ information set $I=\{ a_1, \dots , a_{k_s}\} \subset \{ 1, \dots, k(s+1)\}$ and let $t_s$ be the maximum number of errors that the code can tolerate within the time interval $[a_1, \dots , a_{k_s}]$. Let $P(k, n,t)$ be the probability of determining the secret message when using a $(n,k)$-block code with error-correcting capability $t$ in the classical McEliece PKC. Then, it follows that the probability of recovering $\widehat{\uu}$ in (\ref{eq:011}) is
$$
\frac{1}{q^{k(s+1)-k_s}} P(k_s, n(s+1), t_s).
$$

\section*{Acknowledgement}

This work was supported by the Portuguese Foundation for Science and Technology (FCT-Funda\c{c}\~{a}o para a Ci\^{e}ncia e a Tecnologia), through CIDMA - Center for Research and Development in Mathematics and Applications, within project UID/MAT/04106/2013.

\section*{References}

\bibliographystyle{elsarticle-harv}
\bibliography{biblio_com_tudo}

\end{document}